\documentclass[11pt]{article}
\usepackage{authblk}
\usepackage{caption}
\usepackage[caption=false]{subfig}

\usepackage{amsmath,amsfonts,amsthm,amssymb,graphics,graphicx,mathtools,epstopdf,etoolbox,indentfirst,enumitem,mleftright,xcolor,booktabs,footnote,relsize,thmtools}

\usepackage[scr=euler]{mathalpha} 

\makeatletter
\renewcommand{\paragraph}{%
  \@startsection{paragraph}{4}%
  {\z@}{1.25ex \@plus 1ex \@minus .2ex}{-1em}%
  {\normalfont\normalsize\bfseries}%
}
\makeatother

\let\originalleft\left
\let\originalright\right
\renewcommand{\left}{\mathopen{}\mathclose\bgroup\originalleft}
\renewcommand{\right}{\aftergroup\egroup\originalright}

\usepackage[
backend=bibtex,
style=alphabetic,
giveninits=true, 
maxbibnames=99,
minalphanames=3,
date=year
]{biblatex}
\usepackage{hyperref} 
\usepackage{doi} 

\renewbibmacro{in:}{}

\AtEveryBibitem{\clearfield{issue}} 

\newbibmacro{string+doiurl}[1]{%
	\iffieldundef{doi}{%
		\iffieldundef{url}{
			#1%
		}{%
			\href{\thefield{url}}{#1}%
		}%
	}{%
		\href{http://doi.org/\thefield{doi}}{#1}%
	}%
}
\DeclareFieldFormat{title}{\usebibmacro{string+doiurl}{\mkbibemph{#1}}}
\DeclareFieldFormat[article,thesis,incollection,inproceedings,manual,inbook]{title}%
{\usebibmacro{string+doiurl}
	{#1} 
}
\usepackage{algorithm} 
\usepackage[noend]{algpseudocode}
\floatname{algorithm}{Protocol} 
\algrenewcommand\alglinenumber[1]{\normalsize #1.} 

\newcounter{algsubstate}




\definecolor{darkmagenta}{rgb}{0.85, 0, 0.45}
\hypersetup{
colorlinks = true,
citecolor= blue,
urlcolor= blue,
linkcolor = darkmagenta
}

\newcommand{\ket}[1]{\left| #1 \right>}
\newcommand{\bra}[1]{\left< #1 \right|}
\newcommand{\ketbra}[2]{\ket{#1} \!\! \bra{#2}}
\newcommand{\pure}[1]{\ketbra{#1}{#1}}

\newcommand{\defvar}{\coloneqq} 
\newcommand{\eps}{\epsilon}

\newcommand{\idmap}{\mathcal{I}} 

\newcommand{\norm}[1]{\left\lVert#1\right\rVert} 
\newcommand{\suchthat}{\text{ s.t.}} 
\newcommand{\term}[1]{\textbf{#1}}

\newcommand{\esecure}{\eps^\mathrm{QKD}}
\newcommand{\ecert}{\eps^\mathrm{cert}}
\newcommand{\etotal}{\eps^\mathrm{total}}

\newcommand{\symvec}[1]{\boldsymbol{#1}} 
\newcommand{\rv}[1]{\mathtt{#1}} 
\newcommand{\Iconf}{\rv{I}^\mathrm{conf}} 
\newcommand{\regconf}{\rv{S}^\mathrm{conf}} 

\newcommand{\KA}{K^A} 
\newcommand{\KB}{K^B} 
\newcommand{\Eini}{Q^\mathrm{initial}}
\newcommand{\Efin}{E^\mathrm{final}}
\newcommand{\atk}{\mathscr{A}} 
\newcommand{\mdl}{\mathscr{M}} 
\newcommand{\universe}{\mathfrak{U}_\mathrm{models}} 
\newcommand{\channA}{\mathcal{E}^{\atk}} 
\newcommand{\channAM}{\mathcal{E}^{\atk,\mdl}} 
\newcommand{\channF}{\mathcal{F}} 
\newcommand{\channAMF}{\overline{\mathcal{E}}^{\atk,\mdl}} 
\newcommand{\Nmax}{N_{\mathrm{max}}} 
\newcommand{\Nparam}{N_{\mathrm{param}}} 
\newcommand{\param}{\mu}
\newcommand{\robustset}{S_\mathrm{robust}} 
\newcommand{\nChar}{n_{\textrm{char}}} 

\newcommand{\xdark}{\param_\mathrm{dark}}
\newcommand{\Idark}{[\xdark^\mathrm{low},\xdark^\mathrm{upp}]}

\newtheorem{remark}{Remark}[section]

\newtheorem{proposition}{Proposition}[section]
\newtheorem{definition}{Definition}[section]
\newtheorem{condition}{Condition}[section]

\newfloat{process}{htbp}{loa}
\floatname{process}{Process} 

\newcommand{\markup}[1]{{\color{black}#1}}

\begin{document}

\title{\textbf{Incorporating device characterization into security proofs}}
\renewcommand\Affilfont{\itshape\small} 

\author[1]{Ernest Y.-Z.\ Tan}
\author[1]{Shlok Nahar}
\affil[1]{Institute for Quantum Computing and Department of Physics and Astronomy, University of Waterloo, Waterloo, Ontario N2L 3G1, Canada.}

\date{}

\maketitle

\begin{abstract}
Typical security proofs for quantum key distribution (QKD) rely on having some model for the devices, with the security guarantees implicitly relying on the values of various parameters of the model, such as dark count rates or detector efficiencies. Hence to deploy QKD in practice, we must establish how to certify or characterize the model parameters of a manufacturer's QKD devices. We present a rigorous framework for analyzing such procedures, laying out concrete requirements for both the security proofs and the certification or characterization procedures. In doing so, we describe various forms of conclusions that can and cannot be validly drawn from such procedures, addressing some potential misconceptions. We also discuss connections to composable security frameworks and some technical aspects that remain to be resolved in that direction.
\end{abstract}

\section{Introduction}

Quantum cryptography is a promising prospect for applications of quantum technology, allowing for the realization of cryptographic tasks that would be classically impossible. For instance, a prominent form of quantum cryptography is quantum key distribution (QKD), which distributes secret keys with information-theoretic security between two honest parties Alice and Bob, even in the presence of an eavesdropper Eve. Accordingly, security proofs have been developed for a wide variety of protocols in quantum cryptography, aiming to ensure that they can indeed be safely deployed. For this work, we shall focus our discussion mainly on QKD as a central example, though we note that the overall principles should broadly generalize to many other classes of protocols.

Importantly, such security proofs are only valid for particular classes of models for the device used to implement the corresponding protocol. (For QKD protocols in particular, it might be more accurate to refer to a \emph{pair} of devices held by Alice and Bob, but for brevity in this work we shall often simply refer to it as a single device.)
Given this, we should account for the fact that real-life devices would usually have various imperfections as compared to some idealized ``perfect'' behaviour. If these imperfections are sufficiently extreme, lying outside the class of device models covered by the security proof, they might even cause a protocol performed with such a device to be completely insecure. As an example, consider a decoy-state QKD protocol with an imperfectly phase-randomized laser as described in \cite[Eq. (5)]{NUL23}, focusing on the high-loss regime.
If the phase randomization is perfect or sufficiently close to perfect, the protocol would securely yield non-zero key rate even in the high-loss regime.
However, in the ``maximally imperfect'' phase-coherent case, Alice sends out linearly independent states, thus rendering the protocol completely insecure against an unambiguous state discrimination attack at high loss. 

Hence it is important to introduce some process by which we could try to prevent deployment of QKD devices that do not satisfy the model assumptions required for the security proofs to hold.
More specifically, we could aim to achieve this via some notion of a \term{certification procedure}, which either approves the device for use in QKD protocols, or otherwise rejects the device.\footnote{We choose to use the terms ``approve''/``reject'' in order to avoid confusion with the ``accept''/``abort'' decisions in QKD protocols themselves.} In order for this certification procedure to be practical, we would want to phrase it in terms of a tractable number of \term{parameters} that capture ``all relevant device properties'' for a QKD security proof. (We discuss various considerations for choosing these parameters in the main body of this work.) Informally, we expect that the certification procedure should bound the values of these parameters, and then accept or reject the device based on whether these bounds seem ``good enough'' for us to prove the resulting QKD protocols are secure. While this has some similarities to the notions of ``parameter estimation'' or ``acceptance testing'' in QKD protocols~\cite{rennerthesis}, the conceptual distinction here is that it involves the device that will be used by the honest parties Alice and Bob, rather than the operations that Eve performs on states transmitted during the QKD protocol. 

\subsection{Fundamental obstacles in formalizing certification properties}

If a certification procedure could \emph{guarantee} a bound on some parameter with ``no probabilistic uncertainty'' (put another way, if it gives a bound that holds with probability $1$), then there would be no issues with simply saying that the certification procedure only approves a device if the bound on that parameter falls within some range for which we can prove security of the resulting QKD protocols. 
Unfortunately, demanding such a guarantee is usually too stringent a requirement when considering most types of device parameters, especially in the context of quantum devices that have some ``inherent'' random behaviour. 

More realistically, one would usually only be able to make the weaker statement of constructing a \term{confidence interval} for some parameter, which contains the true value of the parameter with high probability (we defer a formal discussion of confidence intervals to the main body of this work). 
In such circumstances, it is still reasonable to suppose that the certification procedure approves the device whenever the confidence interval falls within some suitable range of values.
Critically, however, here there is some small but nonzero probability that the confidence interval fails to correctly capture the true parameter value, and a full security framework should rigorously account for this fact in some form.

It is tempting to suggest that in the presence of such probabilistic uncertainty, one could simply claim that \emph{conditioned} on the certification procedure accepting, with high probability the QKD protocols using the accepted device will be secure. Unfortunately, such a claim is fundamentally impossible in general --- a very similar issue arises regarding the standard security criterion for QKD, when discussing the states conditioned on the protocol accepting; see~\cite[Sec.~4.3]{arx_TTN+25}. To see this, consider the case where the device undergoing the certification procedure is genuinely insecure (in the sense that any QKD keys it produces can be perfectly guessed by Eve), for instance the aforementioned example of a laser emitting fully phase-coherent states. Due to the probabilistic nature of the certification procedure, there will be a nonzero probability of approving such a device --- for a well-designed certification procedure, this probability would presumably be extremely small, but it is not literally zero. Importantly, the device is \emph{still} fundamentally insecure even when conditioned on the certification procedure approving; it is not valid to claim that it is ``secure with high probability conditioned on approval''.

Put another way, such a claim would be a standard mistake regarding conflation of conditional probabilities: even given that the probability of approval conditioned on the device being insecure is small, this does \emph{not} imply that the probability of the device being insecure conditioned on approval is small. Connecting these probabilities rigorously would require an explicit choice of a \term{Bayesian prior} regarding the device behaviour, which does not feature within the usual frameworks of quantum cryptography.

\subsection{Our contributions}

In light of the above, a more careful treatment of probabilistic uncertainty in the certification procedure is necessary. This is the main purpose of this work: we lay out a framework via which this issue can be properly addressed, and establish rigorous statements regarding the final states produced when a QKD device is subject to such a certification procedure followed by multiple QKD protocol instances (in the event of approval). In the process, we obtain clear guidelines for what is required from both the certification procedures and the QKD security proofs, which we informally summarize as follows:
\begin{enumerate}
\item The certification procedure should construct confidence intervals for all parameters ``relevant to'' a security proof, and reject the device if any of these confidence intervals are not contained within some pre-designated ranges.
\item The security proofs should have some ``robustness'' to imperfections, in the sense that they must ensure security holds for all devices with parameters lying within some nontrivial ranges, rather than only idealized devices with some ``perfect'' values for those parameters.
\end{enumerate}
We give a more detailed and rigorous description of these criteria in the main body of this work, particularly Proposition~\ref{prop:mainbound} and Sec.~\ref{sec:certify}--\ref{sec:securityproof}.
As long as these criteria are satisfied, we can obtain useful conclusions about the states produced by the QKD protocol instances, though we re-emphasize that due to the fundamental obstacles mentioned above, it is  \emph{not} valid to suppose they yield statements of the form (for instance) ``conditioned on the certification procedure approving, the probability that Eve can guess the final keys is small''. 

This paper is structured as follows. In Sec.~\ref{sec:framework}, we lay out our framework for describing certification procedures and QKD protocols, and present our main result (Proposition~\ref{prop:mainbound}) regarding rigorous statements that can be made for the states produced at the end of such a process. In Sec.~\ref{sec:certify}, we discuss what this implies in terms of what is required in the certification procedure, and conversely, in Sec.~\ref{sec:securityproof} we discuss what is required from the QKD security proofs. \markup{In Sec.~\ref{sec:concreteExample}, we provide a pedagogical example that illustrates the application of our work to characterizing dark count rates in a QKD detection setup.} Finally, in Sec.~\ref{sec:conclusion} we summarize our findings. Some technical details and variants are deferred to the appendices, such as the possibilities of ``adaptive'' protocol choices based on characterizing device parameters (Appendix~\ref{app:varproof}) and connections to the framework of Abstract Cryptography (Appendix~\ref{app:AC}). 

\section{Proposed framework}
\label{sec:framework}

We will use $d(\rho,\sigma) \defvar \frac{1}{2}\norm{\rho-\sigma}_1$ to denote the \term{trace distance} between any two states $\rho,\sigma$, as described in e.g.~\cite{NC10}. 
All states should be understood to be normalized states, including when conditioned on some event. 
In this work, to denote a quantum channel $\mathcal{E}$ from register $Q$ to register $Q'$, we will often use the abbreviated notation
$\mathcal{E}: Q \to Q'$
rather than writing out the formal statement of it being a completely positive trace-preserving map~\cite{NC10} from states on $Q$ to states on $Q'$. 
We will also leave tensor products with identity channels implicit; e.g.~for a channel $\mathcal{E}:Q\to Q'$, we often use the compact notation
$\mathcal{E}[\rho_{QR}] \defvar (\mathcal{E} \otimes \idmap_R)[\rho_{QR}]$ where $\idmap_R$ is an identity channel on $R$.

\subsection{A prelude: existing QKD security definition}

Currently, many QKD security proofs are centered around showing that a particular trace-distance security definition holds, as discussed in e.g.~\cite{PR22,BHL+05,HT12}. This security definition can be formulated in a number of equivalent ways; here we shall present a formulation convenient for our work. 

Specifically, first suppose at the end of a QKD protocol, let $\KA$ and $\KB$ denote the final keys obtained by Alice and Bob respectively, and let $\Efin$ denote all registers that may potentially be available to an adversary in the future. We shall allow the length of the final keys $\KA \KB$ to be picked during the protocol, i.e.~we allow for what have been called \term{variable-length} or \term{adaptive-length} protocols~\cite{PR22,BHL+05,HT12}. While many existing QKD proofs have instead focused on \term{fixed-length} protocols that simply make an accept/abort decision and output keys of fixed length when they accept, these can be viewed as just a special case of variable-length protocols by viewing the abort outcome as  a key of zero length.

Given the above register notation, we moreover need a formalism to describe possible attacks by Eve. For the purposes of this work, we shall take a fairly ``minimalist'' perspective that should apply to a broad range of formalisms --- we only require that for any particular attack $\atk$, there exists a corresponding quantum channel $\channA:\Eini \to \KA \KB \Efin$ describing how that attack transforms any state on some initial register $\Eini$, describing ``all relevant registers'' at the start of the protocol, to some final state on $\KA \KB \Efin$, where $\Efin$ denotes all registers Eve holds at the end of the protocol. We emphasize we do not place any restrictions on the ``type of mathematical object'' $\atk$ is formalized as, leaving it up to the relevant use case (it could for instance be a channel, a quantum comb~\cite{CDP09}, a {causal box}~\cite{PMM+17}, or perhaps some even more general structure), other than this fairly minimal requirement. {In particular, we are not necessarily \emph{identifying} the attacks $\atk$ with the channels $\channA$, in that e.g.~we do not require the mapping $\atk \to \channA$ to be bijective --- this allows $\atk$ to potentially be an element of some much ``larger'' set or class than the set of quantum channels.} Note that the channel $\channA$ is intended to describe ``all physical processes'' during the protocol, including actions by honest parties such as state preparations and measurements for prepare-and-measure protocols.

With this in mind, the standard trace-distance security definition can be phrased as saying that a QKD protocol is $\esecure$-secure (also sometimes called $\esecure$-sound~\cite{PR22}) if for any attack $\atk$ and any initial state $\sigma$, we have
\begin{align}\label{eq:standarddefn}
d\left( \channA\left[\sigma_{\Eini}\right] \,,\, \mathcal{K} \circ \channA\left[\sigma_{\Eini}\right] \right) \leq \esecure,
\end{align}
where $\mathcal{K}$ is a channel that reads the length $\ell$ of the keys on registers $\KA \KB$ and replaces them with perfectly correlated uniform keys \emph{of that length}, i.e.~the state 
\begin{align}
\omega^\ell_{\KA\KB} \defvar \sum_{k\in\{0,1\}^\ell} 2^{-\ell} \ketbra{kk}{kk}_{\KA\KB}.
\end{align}
The above security definition is sometimes presented with an additional purifying register than the channel $\channA$ does not act on, viewing this additional purification as being held by Eve. For the purposes of this work, however, we shall view any such purification as implicitly being included in the domain of the channel $\channA$, simply requiring that $\channA$ acts as the identity on that register.

\begin{remark}\label{remark:misconception}
Informally, the state $\mathcal{K} \circ \channA\left[\sigma_{\Eini}\right]$ in Eq.~\eqref{eq:standarddefn} is a somewhat ``ideal'' state, and the security criterion states that the real state produced at the end of the protocol is $\esecure$-close to this ``ideal'' state. To avoid possible misconceptions, we briefly emphasize that this ``ideal'' state is \emph{not} of the form $\sum_{k\in\{0,1\}^\ell} 2^{-\ell} \ketbra{kk}{kk}_{\KA\KB} \otimes \rho^{(\ell)}_{\Efin}$ for any specific length $\ell$ (and some state $\rho^{(\ell)}_{\Efin}$), but rather it is a mixture of such states, with the weights in the mixture being generally \emph{unknown} (as they depend on Eve's arbitrary attack). In particular, this means that for fixed-length protocols, it would \emph{not} be correct to say that the real state conditioned on accepting is $\esecure$-close to a state of that form, since the security definition only involves a mixture of the states conditioned on accepting and aborting. We emphasize however that as discussed in~\cite{PR22}, in the context of the standard QKD assumptions, this definition still suffices to yield composable security.
\end{remark}

We now turn to the question of how to slightly modify this security definition in the context of device imperfections and certification. 

\subsection{Security with respect to device models}

In our above discussion of the security definition, in order to obtain a channel $\channA$ from the attack $\atk$, we in fact implicitly needed to rely on having some model for the QKD device used in the protocol. This works well enough under a ``basic'' QKD security framework where the device (or at least all its ``relevant'' properties) is assumed to be exactly characterized, but in the presence of device imperfections or inexact characterization, it is helpful to make this dependence on the device model more explicit. Hence we shall now slightly rephrase the security definition in a manner more suited for this context, which we shall use for this work.

Consider a QKD device that is about to undergo a certification procedure, intended to certify it for use in up to $\Nmax$ instances of some QKD protocol (after which it is discarded, or sent for re-certification).\footnote{More generally, one might want to consider the possibility of some form of ``mid-lifespan'' certification of the QKD device. However, we leave a detailed analysis of this to be addressed in future work --- for this work, we shall simply take the perspective that this could potentially be viewed as a ``fresh'' certification of the device for another number of QKD protocol instances.} We shall suppose that the device behaviour across all those $\Nmax$ instances is described by some \term{device model}, which we shall denote as $\mdl$. The device model $\mdl$ is intended to fully capture the behaviour of the device throughout those $\Nmax$ protocol instances, in the following sense. Suppose that at the end of each instance, we let $\KA_j \KB_j$ denote the keys produced by that instance, and we let $E_j$ denote \emph{all other registers} that might possibly be accessible to Eve in any form in the future (there are some subtleties in interpreting this, which we defer to Remark~\ref{remark:memory} later).  Similar to the earlier discussion, let us use $\atk$ to denote some arbitrary attack by Eve across all $\Nmax$ instances.\footnote{An alternative formalism could be to specify one attack $\atk_j$ for each instance, but that can simply be considered a special case of the formalism we describe here, for instance by setting $\atk = (\atk_1, \atk_2, \dots)$.} Once again, we do not impose restrictions on the exact ``types of mathematical objects'' used to formalize the device model $\mdl$ and the attack $\atk$. Our central requirement is only the following: 
\begin{condition}\label{cond:channels}
Given any model $\mdl$ and any attack $\atk$, there exists a corresponding sequence of channels $\channAM_j:Q_{j-1} \to \KA_j \KB_j E_j$, describing the physical transformations induced by the $j^\text{th}$ QKD instance from states on some ``initial'' register $Q_{j-1}$ to states on $\KA_j \KB_j E_j$ . 
\end{condition}

Under this framework, we introduce the following definition of what it means for a sequence of QKD protocol instances to be secure for a given device model $\mdl$:
\begin{definition}\label{def:modelsecure}
Following the above description, consider a device for use in up to $\Nmax$ instances of a QKD protocol.
For each $j\in\{1,2,\dots,\Nmax\}$, we say that the $j^\text{th}$ QKD protocol instance is \term{$\esecure_j$-secure under device model $\mdl$} if for any attack $\atk$ and any initial states $\sigma^{(j)}$, we have:
\begin{align}\label{eq:epsQKDj}
d\left( \channAM_j\left[\sigma^{(j)}_{Q_{j-1}}\right] \,,\, \mathcal{K}_j \circ \channAM_j\left[\sigma^{(j)}_{Q_{j-1}}\right] \right) \leq \esecure_j,
\end{align}
where $\mathcal{K}_j$ is a channel that reads the length $\ell$ of the keys on registers $\KA_j \KB_j$ and replaces them with perfectly correlated uniform keys \emph{of that length}, i.e.~the state 
\begin{align}\label{eq:idealkeyj}
\omega^\ell_{\KA_j\KB_j} \defvar \sum_{k\in\{0,1\}^\ell} 2^{-\ell} \ketbra{kk}{kk}_{\KA_j\KB_j}.
\end{align}
\end{definition}

Observe that in a context where the device is used sequentially for multiple instances of a QKD protocol, we can write any state generated after $n$ such instances in the following form (as long as certain technicalities hold regarding memory registers; see Remark~\ref{remark:memory}):
\begin{align}
\channAM_n \circ \channF_{n-1} \circ \dots \channAM_2 \circ \channF_1 \circ \channAM_1 \left[\sigma^{(0)}_{Q_0}\right], \label{eq:state_n}
\end{align}
for some initial state $\sigma^{(0)}$ and some channels $\channF_j: \KA_j \KB_j E_j \to Q_j$ that represent arbitrary ``physical processes'' taking place between the completion of the $j^\text{th}$ protocol instance and the following instance. In the subsequent discussions, we will be discussing properties of states with the form in~\eqref{eq:state_n}, with the understanding that they are meant to describe the state at the end of $n$ QKD protocol instances.

We highlight that in this formalism we have chosen, we allow each of the $\channF_j$ channels to potentially act nontrivially on the keys $\KA_{j}\KB_{j}$ produced by the preceding channel, not just the ``side-information'' register $E_{j}$. This is because these channels are intended to capture ``all physical processes'' that take place in each instance, and in many practical applications of QKD, one would presumably want to make use of the keys produced in each instance before the next instance has fully completed. This could potentially leak information about those keys to an adversary --- as an example, if $\KA_{j}\KB_{j}$ are used in a one-time-pad protocol, then the ciphertext can be nontrivially correlated to $\KA_{j}\KB_{j}$ (this being essentially the reason one-time-pads cannot be used to encrypt more than one message). Another particularly critical example would be using part of the keys produced by one instance to authenticate the subsequent instance, as discussed in e.g.~\cite{Por14}. By including these keys in the domain of $\channF_j$, our formalism is able to capture such scenarios as well; in particular, we do not require that the keys remain ``unused'' until the devices have reached the end of their certification lifespan (which would clearly be impractical).

\begin{remark}\label{remark:memory}
Some care is needed in interpreting the registers $E_j$ in the above definition. Specifically, we first emphasize it is known that if the behaviour of the device in some QKD protocol instance depends nontrivially on some ``memory'' of the raw key retained from previous instances, then this fundamentally breaks the composable security of QKD~\cite{BCK13}. {While that work was focused on device-independent QKD, the same issue arises for any other form of QKD (including measurement-device-independent QKD) in the absence of assumptions ruling out such forms of memory effects in the devices, whether they be for state preparation or measurement.}
Therefore, we should consider how this is to be reflected in the above formalism, denoting the memory retained in a QKD device after each protocol instance as $R_j$. 

One option is to view $R_j$ as being implicitly part of the $E_j$ register, but in that case showing that the bound~\eqref{eq:epsQKDj} holds would require proving that the memory is \emph{also} basically ``decoupled'' from the secret key --- with this, we would have no interpretational difficulties with the modelling and theoretical framework; however, proving that the memory is indeed decoupled from the key would involve invoking suitable physical assumptions. Alternatively, we could exclude $R_j$ from the $E_j$ register, but note that in the above formalism this is tantamount to saying that \emph{no ``allowed'' physical process} gives Eve access to $R_j$ in any form, in that none of the subsequent $\channF_j$ and $\channAM_j$ channels can act on $R_j$ in any way --- for practical purposes, this might be operationally equivalent to supposing the devices do not in fact retain memory between protocol instances. 

While these options might appear to be somewhat restrictive conditions, we highlight that the existence of the~\cite{BCK13} attack demonstrates that QKD would \emph{genuinely} become insecure in the absence of assumptions ruling out ``nontrivial'' memory effects, and hence some assumptions of this nature are in fact necessary (unless countermeasures such as~\cite{MS14,CL19,LRR19} are applied, but we leave for future work an analysis of how they interact with a certification framework).
\end{remark}

\subsection{Main result}
\label{subsec:mainresult}

With this formalism, we now turn to the main question of how to analyze certification procedures in these terms. We shall suppose that given some device with model $\mdl$, the certification procedure either \term{approves} the device for use in up to $\Nmax$ QKD protocol instances, in which case it outputs a classical register $F=1$, or it \term{rejects} the device, in which case it outputs $F=0$ (with the probabilities of each outcome being some function of $\mdl$). In that case, the relevant state to consider before the start of the first QKD protocol instance can be written in the form 
\begin{align}\label{eq:initialstate}
\sigma_{Q_0 F} = \Pr[F=0]\, \sigma_{Q_0 | F=0} \otimes \pure{0}_F + \Pr[F=1]\, \sigma_{Q_0 | F=1} \otimes \pure{1}_F, 
\end{align}
where $Q_0$ is the input register to the first channel $\channAM_1$ discussed in Condition~\ref{cond:channels}, and $\sigma_{Q_0 | F=0}$ and $\sigma_{Q_0 | F=1}$ respectively denote the states on this register conditioned on $F=0$ or $F=1$. 

We highlight that in order to be able to draw any nontrivial conclusions in a certification framework, it is necessary to implicitly restrict the ``universe of possible device models'' in some sense. For instance, if we were to consider the hypothetical possibility of a device that behaves \emph{exactly identically in every way} to some ``ideal'' honest device during the certification procedure, only to later behave in some malicious fashion when used in the actual QKD protocols, then it would be impossible for a certification procedure to make nontrivial conclusions about this device (since by hypothesis, none of its malicious behaviour can be ``observed'' during certification). Therefore, it is necessary to rule out at least {some} forms of device behaviour purely \emph{by assumption}, rather than relying on a certification procedure to do so. We shall formalize this in this work by supposing that all device models $\mdl$ lie within some universal set (or class) denoted as $\universe$, which is chosen to be broad enough to capture all ``reasonable'' device behaviours, but still sufficiently specific to allow proving nontrivial statements such as the criteria in Proposition~\ref{prop:mainbound} below.

\begin{remark}
As a simple example, we could for instance restrict $\universe$ to only consist of device models such that the state preparation device and the measurement device behave in an independent and identically distributed (IID) manner every time they are used, even if we do not exactly know the state prepared or measurement performed in a single usage. Note that this is a milder assumption than assuming a restriction to ``IID collective attacks'' (see e.g.~\cite{PR22}) in QKD security proofs; we return to this point at the end of Sec.~\ref{sec:certify}, when discussing confidence intervals. Qualitatively, this rules out the above hypothetical of the device performing very differently during the certification procedure compared to the actual protocol, and allows us to obtain nontrivial conclusions from Proposition~\ref{prop:mainbound} later. However, we emphasize that this is just a starting example, and one could consider more general choices of $\universe$ if desired, as we elaborate on at the end of Sec.~\ref{sec:certify}.
\end{remark}

Since a device model $\mdl$ could in general be a very complicated object (as it should describe the entire behaviour of the device over all $\Nmax$ protocol instances), we shall introduce an additional concept to simplify the subsequent discussion and the practical certification procedure. Namely, we shall suppose that each device model $\mdl$ is associated to a list of $\Nparam$ \term{parameters} $(\param_1, \param_2, \dots, \param_{\Nparam}) \eqqcolon \symvec{\param}$, which informally speaking summarizes some ``main'' properties relevant to a security analysis. As a basic example, one of these parameters could for instance be an upper bound\footnote{In this example we choose to specify the parameter to be an \emph{upper bound} on the dark count rate rather than ``the'' dark count rate, since the latter may not be a well-defined single number if the detector behaviour is not IID.} on the dark count rate of the detectors, where the bound is to hold over all $\Nmax$ protocol instances. 
Again, this is just a simple example, which is not meant to be definitive --- for now we retain flexibility in what can be considered a parameter of the device model, returning to this question later after presenting a notion of what we require from this concept.

\begin{remark}
Note that the simplification in terms of parameters is not strictly mathematically necessary for our subsequent results, which could in principle be rephrased directly in terms of the device models $\mdl$ themselves. However, for practical purposes it is much easier to discuss the certification procedure in terms of a concise list of parameters $\symvec{\param}$.
\end{remark}

With this, we present our main technical statement regarding device certification. We highlight that it formally holds for any choice of the universe of models $\universe$; however, an ``overly broad'' choice of $\universe$ might make it impossible for the criteria in this proposition to be satisfied, as discussed above, making this statement vacuous in those circumstances. Also as discussed above, the channels $\channF_j$ in the statement should be understood to represent any arbitrary ``physical processes'' taking place between the QKD instances.

\begin{restatable}{proposition}{mainbound}\label{prop:mainbound}
    Following the above description, consider a device for use in up to $\Nmax$ instances of a QKD protocol, and let $\robustset$ be some set of parameter values $\symvec{\param}$. Suppose the following criteria are satisfied:
    \begin{enumerate}
    \item \label{crit:cert} For every device model $\mdl \in \universe$ with parameters $\symvec{\param}\notin\robustset$, the certification procedure rejects with probability at least $1-\ecert$, i.e.~$\Pr[F=1] \leq \ecert$. 
    \item \label{crit:QKD} For every device model $\mdl \in \universe$ with parameters $\symvec{\param}\in\robustset$, the $j^\text{th}$ QKD protocol is $\esecure_j$-secure under device model $\mdl$ (Definition~\ref{def:modelsecure}).
    \end{enumerate}
    Then for any $n \in\{1,2,\dots,\Nmax\}$, any state $\sigma_{Q_0 F}$ (as in~\eqref{eq:initialstate}) at the start of the first protocol instance, any attack $\atk$ and device model $\mdl \in \universe$, and any sequence of channels $\channF_j: \KA_j \KB_j E_j \to Q_j$ for $j\in\{1,2,\dots,n-1\}$, we have:
    \begin{align}
    \begin{gathered}
    \Pr[F=1]\; d\left( \channAMF_n \circ \dots \channAMF_1 \left[\sigma_{Q_0 | F=1}\right] \,,\, \mathcal{K}_n \circ \channAMF_n \circ \dots \mathcal{K}_1 \circ \channAMF_1 \left[\sigma_{Q_0 | F=1}\right] \right)\leq \etotal_n, \\
    \text{with } \etotal_n = 
    \ecert+ \sum_{j=1}^n \esecure_j ,
    \end{gathered}
    \label{eq:mainbound}
    \end{align}
    where for notational compactness we write $\channAMF_j \defvar \channAM_j \circ \channF_{j-1}$ (defining $\channF_0$ to be the identity channel on $Q_0$),
    and $\mathcal{K}_j$ is a channel that reads the length $\ell$ of the keys on registers $\KA_j \KB_j$ and replaces them with perfectly correlated uniform keys \emph{of that length}, i.e.~the state $\omega^\ell$ in~\eqref{eq:idealkeyj}.
\end{restatable}
The proof is fairly straightforward, hence we defer it to Appendix~\ref{app:mainproof}.\footnote{In fact, from the proof it can be seen that the upper bound in~\eqref{eq:mainbound} can be sharpened to $\etotal_n = \max\left\{\ecert \,,\, \sum_{j=1}^n \esecure\right\}$.} For now we simply discuss its qualitative implications. In particular, observe that it allows us to more precisely specify some goals we should aim for in security proofs and a certification procedure, via the two criteria. Specifically, the second criterion tells us we should aim to construct security proofs that have some ``robustness'' to device imperfections, in the sense that the QKD protocols are secure for all devices with model parameters in some choice of set $\robustset$. From the other direction, the first criterion tells us the certification procedure should aim to reject all devices with model parameters outside $\robustset$ with high probability. This is the main proposal of our work, and in Sec.~\ref{sec:certify}--\ref{sec:securityproof}, we explore in greater detail how these goals could be achieved in practice.

Before proceeding, we briefly discuss how to interpret the trace-distance bound obtained from the above result. In particular, observe that the state $\channAMF_n \circ \dots \channAMF_1 \left[\sigma_{Q_0 | F=1}\right]$ in the above formula is basically the real state produced after $n$ protocol instances for a given attack and device model (conditioned on the certification procedure approving), while the state $\mathcal{K}_n \circ \channAMF_n \circ \dots \mathcal{K}_1 \circ \channAMF_1 \left[\sigma_{Q_0 | F=1}\right]$ is essentially a more ``ideal'' version in which the keys produced at the end of each protocol instance have been immediately replaced by perfect keys of that length. Hence our above result is an upper bound on the trace distance between these states, weighted by the prefactor of $\Pr[F=1]$. Note that the possibility of Eve performing different actions depending on the certification outcome can be implicitly incorporated into the framework we describe here, simply by defining the channels $\channAM_n$ to act on a classical register storing that outcome as well.

However, we emphasize (similar to Remark~\ref{remark:misconception}) that it would \emph{not} be correct to view the above bound as saying that the trace distance between the states conditioned on the certification procedure approving (i.e.~$\channAMF_n \circ \dots \channAMF_1 \left[\sigma_{Q_0 | F=1}\right]$ and $\mathcal{K}_n \circ \channAMF_n \circ \dots \mathcal{K}_1 \circ \channAMF_1 \left[\sigma_{Q_0 | F=1}\right]$) is small. (An analogous point also applies for conditioning on some or all of the QKD protocol instances accepting, as already discussed in Remark~\ref{remark:misconception}.) This is due to the prefactor of $\Pr[F=1]$ in~\eqref{eq:mainbound}, which allows the trace distance to potentially be 
large if $\Pr[F=1]$ is small. Still, we can recover some intuitive properties from this bound as follows. Given a state $\sigma_{Q_0 F}$ as in~\eqref{eq:initialstate} at the start of the first protocol instance, we can write the final state at the end of all the instances in the form 
\begin{align}\label{eq:finalrho}
\begin{aligned}
\rho_{\KA_n \KB_n E_n F} &= \Pr[F=0]\, \rho_{\KA_n \KB_n E_n | F=0} \otimes \pure{0}_F \\
&\qquad+  \Pr[F=1]\, \channAMF_n \circ \dots \channAMF_1 \left[\sigma_{Q_0 | F=1}\right] \otimes \pure{1}_F, 
\end{aligned}
\end{align}
i.e.~a mixture of the cases conditioned on the certification procedure approving or rejecting. If we now define another state $\widetilde{\rho}$ via
\begin{align}\label{eq:finalrhoideal}
\begin{aligned}
\widetilde{\rho}_{\KA_n \KB_n E_n F} &= \Pr[F=0]\, \rho_{\KA_n \KB_n E_n | F=0} \otimes \pure{0}_F \\
&\qquad+  \Pr[F=1]\, \mathcal{K}_n \circ \channAMF_n \circ \dots \mathcal{K}_1 \circ \channAMF_1 \left[\sigma_{Q_0 | F=1}\right] \otimes \pure{1}_F, 
\end{aligned}
\end{align}
then the bound~\eqref{eq:mainbound} is equivalent to
\begin{align}
d\left( \rho_{\KA_n \KB_n E_n F} \,,\, \widetilde{\rho}_{\KA_n \KB_n E_n F} \right)\leq \etotal_n,
\end{align}
i.e.~these states are genuinely $\etotal_n$-close in trace distance. 

We stress that the states $\rho_{\KA_n \KB_n E_n F}$ and $\widetilde{\rho}_{\KA_n \KB_n E_n F}$ are \emph{mixtures} of the states conditioned on the certification procedure approving or rejecting, not just the states conditioned on approving. However, in many contexts $\rho_{\KA_n \KB_n E_n F}$ is indeed the final state of interest, since the actual state that would be physically produced is indeed this mixture. As for $\widetilde{\rho}_{\KA_n \KB_n E_n F}$, while it is not as ``ideal'' a state as one might initially have hoped for, it should still possess many useful properties if we presume that the component $\rho_{\KA_n \KB_n E_n | F=0}$ conditioned on the certification procedure rejecting is essentially ``trivial'' in some sense (for instance, the device is never used for any QKD protocols and all the registers $\KA_j \KB_j$ are just set to some trivial ``null'' values). In particular, for readers already familiar with the Abstract Cryptography framework~\cite{MR11,PR22}, we discuss in Appendix~\ref{app:AC} how this analysis suffices to show we have a composably secure construction of a particular resource (appropriately modified from the usual resource constructed by QKD), under certain assumptions. 

For now, we simply note that due to the fundamental properties of trace distance, the trace-distance bound implies for instance that if we were to replace the actual state $\rho_{\KA_n \KB_n E_n F}$ with $\widetilde{\rho}_{\KA_n \KB_n E_n F}$ in any other security analysis, then the probability of any event cannot change by more than $\etotal_n$, which would already be a useful operational result for many security statements (as long as one keeps in mind that $\widetilde{\rho}_{\KA_n \KB_n E_n F}$ has the structure in~\eqref{eq:finalrhoideal} for some ``unknown'' probabilities $\Pr[F=0]$, $\Pr[F=1]$). 

One might question whether the $\Pr[F=1]$ prefactor in the Proposition~\ref{prop:mainbound} bound is truly necessary. We highlight however that it is straightforward to construct counterexamples (even with simple IID devices) such that the final state conditioned on the certification procedure approving satisfies
\begin{align}\label{eq:largedistance}
d\left( \channAMF_n \circ \dots \channAMF_1 \left[\sigma_{Q_0 | F=1}\right] \,,\, \mathcal{K}_n \circ \channAMF_n \circ \dots \mathcal{K}_1 \circ \channAMF_1 \left[\sigma_{Q_0 | F=1}\right] \right) \approx 1,
\end{align}
i.e.~it is very far from the ``ideal'' state. (This does not contradict Proposition~\ref{prop:mainbound}, precisely because the bound in it has the $\Pr[F=1]$ prefactor.)
Specifically, we simply need to again consider a ``genuinely insecure'' device as described in the introduction, such as the example of a fully phase-coherent laser being used for QKD protocols in a high-loss regime. As discussed in the introduction, Eve would always be able to perfectly learn the QKD keys generated using such a device, even conditioned on the certification approving --- hence the trace distance between the states conditioned on approving, as in~\eqref{eq:largedistance}, would indeed be close to $1$.
Therefore, any attempt to bound the trace distance involving the state conditioned on the certification procedure accepting would necessarily need to account for $\Pr[F=1]$ in some form, even if it might not be strictly in the form presented in~\eqref{eq:mainbound}.

\section{Considerations for the certification procedure}
\label{sec:certify}

In this section, we outline some examples of how a certification procedure could be designed to reject ``bad'' devices, in the sense of those with model parameters outside of some set $\robustset$ for which we can prove the QKD protocols are secure (see Sec.~\ref{sec:securityproof}). We again highlight that the following examples are not necessarily definitive, in the sense that Proposition~\ref{prop:mainbound} is not concerned with the details of exactly how this is achieved --- it only requires that such devices are rejected with high probability. Still, we believe that the scenarios considered here are reasonably reflective of procedures that could be applied in practice; at the end of this section, we also discuss some natural extensions.

To begin, we briefly present the basic definition of a {confidence interval}:
\begin{definition}\label{def:confidence}
Consider a probability distribution described by some fixed value $\param$, usually referred to as a \term{parameter}. If we have two random variables $\rv{X}^\mathrm{low}$ and $\rv{X}^\mathrm{upp}$ such that the interval $\Iconf \defvar [\rv{X}^\mathrm{low}, \rv{X}^\mathrm{upp}]$ contains $\param$ with probability at least $1-\eps$, then we say $\Iconf$ is a \term{confidence interval} for $\param$ at confidence level $\eps$.
\end{definition}
We emphasize a standard conceptual point regarding confidence intervals: in the above definition, $\param$ is to be understood as a \emph{fixed} value, with no ``randomness'' involved. Therefore, when one says that $\Iconf$ contains $\param$ with some probability, the ``random quantity'' is \emph{not} the value $\param$ but rather the interval $\Iconf$, and it would be a misconception to speak of (for instance) the probability distribution of $\param$.
As a default example, suppose we have $n$ random variables $\rv{X}_1, \rv{X}_2, \dots , \rv{X}_n$ following an IID distribution 
$P_{\rv{X}_1 \rv{X}_2 \dots \rv{X}_n} = P_{\rv{X}_1} P_{\rv{X}_2} \dots P_{\rv{X}_n}$,
where the expected value for any individual $\rv{X}_j$ is some value $\param$. A common method to construct a confidence interval for $\param$ is to compute the ``sample mean'' $\rv{X}_\mathrm{avg} \defvar \frac{1}{n} \sum_j \rv{X}_j$, and then set $\Iconf = [\rv{X}_\mathrm{avg} - \delta, \rv{X}_\mathrm{avg} + \delta]$ for some suitably chosen value $\delta$ (depending on $n$, $\eps$ and some other parameters, such as the range of the random variables). Notice that in this process, the value of $\param$ is a \emph{fixed} quantity, with the ``random quantities'' being $\{\rv{X}_j\}$, $\rv{X}_\mathrm{avg}$, and thus $\Iconf$.

Critically, however, the standard valid interpretation of confidence intervals is fully compatible with the framework suggested by Proposition~\ref{prop:mainbound}, as we now describe. As a crude starting demonstration, consider a very simplistic scenario where $\universe$ consists only of devices that behave in an IID manner every time they are used, and the only parameter of importance is, say, the dark count rate $\xdark$ in the photon detection device (which is constant across all uses, by the IID restriction). Suppose furthermore that the set $\robustset$ is simply some (fixed) interval $\Idark$, i.e.~we have established security proofs such that QKD protocols performed with this device are secure whenever the device has dark count rate $\xdark$ within the interval $\Idark$. In that case, if we design the certification procedure such that it constructs a confidence interval $\Iconf_\mathrm{dark}$ for $\xdark$ at confidence level $\ecert$, and set the approval condition to
\begin{align}
\text{Approve if and only if } \Iconf_\mathrm{dark} \subseteq \Idark,
\end{align}
then it follows directly from the definition of confidence intervals that indeed, this certification procedure fulfills the criterion in Proposition~\ref{prop:mainbound} that it rejects with probability at least $1-\ecert$ whenever $\xdark \notin \Idark$. Specifically, simply recall that from Definition~\ref{def:confidence} we have that $\Iconf_\mathrm{dark}$ contains $\xdark$ with probability at least $1-\ecert$; therefore, whenever $\xdark \notin \Idark$, with that probability $\Iconf_\mathrm{dark}$ contains the value $\xdark$ outside $\Idark$, causing the certification to reject. Also notice that as one would intuitively expect, it is advantageous to construct $\Iconf_\mathrm{dark}$ as small as possible, since it makes the certification procedure more likely to approve ``good'' devices with $\xdark \in \Idark$, while still retaining the desired security properties (as long as $\Iconf_\mathrm{dark}$ still satisfies Definition~\ref{def:confidence}). Finally, we acknowledge that a natural question is whether one could relax the approval condition to $\Iconf_\mathrm{dark}$ having non-empty overlap with $\Idark$; however, this is not strong enough to show the desired property.

More realistically, we would want to consider multiple parameters $(\param_1, \param_2, \dots, \param_{\Nparam}) \eqqcolon \symvec{\param}$ as described in the previous section. Suppose for simplicity  that $\robustset$ is simply defined by ``componentwise'' intervals, i.e.~it has the following form (for some fixed values $\param_k^\mathrm{low}, \param_k^\mathrm{upp}$): 
\begin{align}\label{eq:intervalwiseS}
\robustset = \left\{ \symvec{\param} \in \mathbb{R}^{\Nparam} \;\middle|\; \param_k \in [\param_k^\mathrm{low}, \param_k^\mathrm{upp}] \text{ for all } j \right\}.
\end{align}
In that case, a certification procedure satisfying the criterion in Proposition~\ref{prop:mainbound} would be to construct for each $k$ a confidence interval $\Iconf_k$ for $\param_k$ at confidence level $\ecert$, and set the approval condition to
\begin{align}
\text{Approve if and only if } \Iconf_k \subseteq [\param_k^\mathrm{low}, \param_k^\mathrm{upp}] \text{ for all } k.
\end{align}
Note that we do not require the confidence interval constructions to be independent in any way; we only require that each confidence interval individually satisfies Definition~\ref{def:confidence}.
Again, to see that this has the desired property, simply observe that if $\symvec{\param} \notin \robustset$, that means at least one $\param_k$ lies outside the corresponding interval $[\param_k^\mathrm{low}, \param_k^\mathrm{upp}]$. Therefore, with probability at least $1-\ecert$ the corresponding confidence interval $\Iconf_k$ contains this value outside $[\param_k^\mathrm{low}, \param_k^\mathrm{upp}]$, forcing the certification to reject (without any independence assumptions across the confidence intervals). {Again, one might ask whether the approval condition could be relaxed to having $\Iconf_k \subseteq [\param_k^\mathrm{low}, \param_k^\mathrm{upp}]$ for only some $k$ rather than all, or to only requiring non-empty overlap between $\Iconf_k$ and $[\param_k^\mathrm{low}, \param_k^\mathrm{upp}]$ --- however, these versions are again not strong enough to show the desired property.} 

We believe that the above example with multiple parameters should conceptually suffice to match certification procedures that can be performed in practice (though it does clarify precisely what should be set as the approval condition). In particular, we note that the above structure is completely indifferent to the exact method used to construct the confidence intervals, as long as they genuinely satisfy the foundational definition (Definition~\ref{def:confidence}), and hence is very flexible in that regard. One might even wish to consider the more general concept of \term{confidence regions}, which are subsets of $\mathbb{R}^{\Nparam}$ that are not necessarily defined by ``componentwise'' intervals (see Definition~\ref{def:confregion} for more details); similarly, one could consider choices of $\robustset$ that are not of the form~\eqref{eq:intervalwiseS}. Our above analysis generalizes straightforwardly by requiring the confidence region to be completely contained in $\robustset$. 

We close this section with some discussions of points to note in practice.

\paragraph{Assumptions in constructing confidence intervals.} It is worth highlighting that in order to construct a confidence interval in the first place, one usually needs to make some sort of assumption on the device behaviour during the certification procedure, such as an IID or independence property, or perhaps a martingale property (for instance, to apply Azuma's inequality). However, it seems reasonable to suggest that since this is a consideration of the \emph{devices} manufactured for use by the honest parties, one could take a less ``adversarial'' perspective on their behaviour as compared to a QKD security proof analyzing Eve's \emph{attack} on the quantum states being transmitted, and thus reasonably restrict $\universe$ such that the devices have the required independence properties. (In other words, this is weaker than an assumption that Eve is restricted to ``IID collective attacks''~\cite{PR22}.) Furthermore, in any case it seems unlikely that one would be able to draw nontrivial conclusions from certification without at least some assumptions of this nature. 

\paragraph{Non-constant behaviour.} While the above discussion is very general regarding what can be considered a ``parameter'', it may be worth discussing some specific scenarios more closely. For instance, it is natural to consider the possibility of a QKD device where its performance ``degrades'' over time in some fashion. One might wonder whether this would require us to include an entire \emph{continuum} of parameters in our formalism, describing its behaviour at every point in time. However, this would be a rather impractical way of capturing this within our formalism. Rather, we believe a more reasonable approach would be to assume the universe of models $\universe$ is restricted such that for any device model $\mdl$ within this universe, this ``degrading behaviour'' over its entire lifespan can be described in terms of some \emph{single} parameter, or a small number of parameters.

For example, single-photon detectors used in satellite QKD are expected to have time-varying dark count rates due to radiation damage and changes in temperature (see e.g.~\cite{lenart2025comparing}). This time-variation can be studied and considered to be a part of the universe of models $\universe$ --- for instance, we could capture this by restricting $\universe$ to be such that for any model $\mdl\in\universe$, the dark count rate in the $j^\text{th}$ QKD instance can be lower and upper bounded by some functions $f^\mathrm{low} (j,\param_\mathrm{temp})$ and $f^\mathrm{upp} (j,\param_\mathrm{temp})$ respectively, where $\param_\mathrm{temp}$ is the temperature around the detector, which one can aim to construct a confidence interval for during certification. This is mainly a basic example to capture the spirit of the idea; practically, one should give a more sophisticated description of $\universe$ and how the models within it can be described in terms of some parameters. Note that there is freedom to ``trade off'' against the security proofs here, in the sense that if the device behaviour is expected to degrade over time as described by the model $\mdl \in \universe$, then it may be possible to compensate for this in the security proofs for the QKD protocols (i.e.~criterion~\ref{crit:QKD} in Proposition~\ref{prop:mainbound}) by reducing the key lengths in later QKD instances, to ensure that the $j^\text{th}$ instance is indeed $\esecure_j$-secure for models with $\symvec{\param}\in\robustset$. 

Another consideration is whether some aspects of the device behaviour might deviate dramatically from the honest behaviour when subject to some form of adversarial attack, such as a Trojan-horse attack \cite{vakhitov2001large}. In that case, the relevant parameter to assign to it within our formalism would be for instance a \emph{bound} on how far it can deviate from the honest behaviour even when subjected to such attacks; for instance, by using an optical isolator \cite{vakhitov2001large}, one might be able to bound the maximum extent of deviation (under ``reasonable'' assumptions \cite{lucamarini2015practical} about the intensity of light the adversary can send, and so on, which would be captured within our framework by restricting the universe $\universe$). The security proofs would then have to ensure that as long as this parameter lies within $\robustset$, the QKD protocol instances are still secure.

Overall, this highlights the need to ``embed'' assumptions into $\universe$ somewhat carefully, so that we can still describe the models with a tractable number of parameters for certification, while still being able to prove security for models with $\symvec{\param}\in\robustset$. This is something that is pragmatically necessary, as the space of ``all conceivable  device models'' would be unreasonably large, and impossible to design a practical certification procedure for. Practically, this would mean choosing $\universe$ according to the accepted best practices at a given time, such that they indeed capture all device behaviours that are ``realistically'' known to occur (even in the presence of adversarial attacks during the QKD protocol), while still being tractable to analyze. It is also important to ensure that the security proofs for the QKD protocols remain ``up-to-date'' with the choice of $\universe$, so that they indeed ensure the required security properties for models with $\symvec{\param}\in\robustset$. In the next section, we turn to discussing this point.

\markup{\paragraph{Practical challenges in metrology.}
A significant practical challenge remains in the \textit{metrological} quality of the certification procedure. In a realistic deployment, the certification is performed using reference devices (e.g., single-photon detectors, laser sources, etc.) that are themselves subject to uncertainties and imperfections.  Further, in order to accurately capture errors such as misalignment, a realistic certification procedure would involve a complete certification of the entire detector or source module rather than the individual parameters that comprise of it.
Rigorously developing a certification procedure that accounts for these imperfections is a challenging task that we leave to future work.
Our framework highlights that the security of a deployed system relies as much on the quality of these certification procedures as it does on the security proof techniques.}

\section{Considerations for the security proofs}
\label{sec:securityproof}

An important point arising from the preceding discussions is that in order for this framework to be practical, we do genuinely need QKD security proofs that have ``nontrivial robustness'' to imperfections, in the sense that we can show criterion~\ref{crit:QKD} in Proposition~\ref{prop:mainbound} holds for some choice of $\robustset$ that does not consist only of a single point. This is because from the preceding discussion, as long as the certification procedure can only construct confidence intervals of nonzero width, it would never approve the device if $\robustset$ consists only of a single point (or is otherwise ``too small'' for reasonably-sized confidence intervals to be contained within it). Hence it highlights the importance of constructing QKD security proofs in such a way that they have some robustness to device imperfections, as is a subject of much ongoing work~\cite{curras2023securityframework, arqand2024mutual, marwah2024proving, sixto2024quantum, zapatero2021security, curras2023security, tupkary2024phaseerrorrateestimation, kamin2025r, nahar2025imperfect}.

Furthermore, given the preceding discussion of the flexibility in choosing parameters, we see that it is important for security proofs to remain ``in touch'' with the practical reality of what can be measured during a certification procedure. One should choose the parameters and the universe $\universe$ in such a way that there is compatibility between the certification procedures and the security proofs, while still ``realistically'' reflecting the actual behaviour of the devices. It is likely that this will need to involve an ongoing process of updating the models and choices of parameters, as better understanding is gained of device characteristics and potential flaws.

\markup{

\section{Pedagogical example: Detector dark counts}
\label{sec:concreteExample}

To demonstrate the application of our framework more concretely, we consider the characterization of a receiver's device for the BB84 protocol. We focus specifically on the characterization of the detectors' dark count rates, illustrating the application of our work.

\subsection{The physical model}

We consider a standard passive basis-choice detection setup consisting of a basis-choice beam splitter leading to two measurement arms (one for the Z basis, one for the X basis), utilizing a total of $4$ single-photon detectors. We assume that the detection setup couples perfectly to a single spatio-temporal input mode, and all detectors have unit efficiencies. Further, we assume that the repetition rate of the protocol is fixed. This repetition rate allows us to define the probability of a dark count of each detector in a single protocol round. We \textit{do not} assume that this dark count probability is perfectly known. Finally, we assume that the detector behaviour is IID in each round, with no memory effects (again, we emphasize that this is \emph{not} restricting an attacker to collective attacks: it only involves the device held by the receiver).

In short, we assume all devices behave perfectly as per specification, except for the dark count probabilities which must be characterized. This defines our universe of models $\universe$. We associate with each device model $\mdl$ a single parameter $\xdark$ which represents the maximum dark count probability of all detectors.
While we could define the dark count probability in each detector to be an independent parameter in our device model, we find this simpler, single parameter description is sufficient for security proofs \cite{nahar2025imperfect,nahar2026proof}.

\subsection{The certification procedure}

Let us define the 
set $\robustset$ as the subset of all parameter values where all 4 detectors have a dark count probability at most $\xdark^\mathrm{upp}$ (for some value $\xdark^\mathrm{upp}$ that is fixed before the certification procedure takes place):
\begin{align}
    \robustset = \left\{ \xdark \in [0,1] \mid \xdark \leq \xdark^\mathrm{upp} \right\}.
\end{align}

\begin{remark}
We emphasize a basic but crucial point: the value of $\xdark^\mathrm{upp}$ in this example (and thus $\robustset$) is to be fixed \emph{before} the certification procedure begins, and not adjusted afterwards. The choice of this value does not affect the security guarantees (as long as the subsequent QKD protocols are implemented with the key lengths prescribed in Sec.~\ref{subsec:robQKDProofDCs} later); however, in order for the certification procedure to be ``practically useful'', we need it to also accept with reasonably high probability whenever 
the device behaves according to the expected ``correct'' behaviour. In particular, this means that we would usually need to choose $\xdark^\mathrm{upp}$ to be somewhat higher than the expected ``correct'' value --- we stress however that the corresponding ``margin of tolerance'' is a \emph{distinct} concept from the value $\delta$ we describe below (which cannot be arbitrarily chosen without affecting the security guarantees), and should not be confused with or merged into it. 
\end{remark}

The certification procedure needs to be designed such that if $\xdark \notin \robustset$, the certification procedure rejects with probability at least $1-\ecert$ as described in criterion~\ref{crit:cert} of Proposition~\ref{prop:mainbound}.
Such a procedure can be implemented by formally integrating finite-size statistical analysis into existing methods \cite{etsi} for characterizing dark count probabilities of individual detectors. We sketch out the main steps in such a procedure:
\begin{enumerate}
    \item \textbf{Data collection:} Each detector is individually isolated from any input light. The detectors are each operated for $\nChar$ time bins. We record the resulting vector of click counts $\mathbf{k} = (k_1, k_2, k_3, k_4)$ for each of the four detectors, and define a corresponding vector of click frequencies
    \begin{align}
        \mathbf{f} = (f_1, f_2, f_3, f_4) \defvar \left( \frac{k_1}{\nChar} , \frac{k_2}{\nChar} , \frac{k_3}{\nChar} , \frac{k_4}{\nChar} , \right) .
    \end{align}
    More generally, one could use a different number of time bins for each detector, but it is not hard to modify the analysis accordingly if that is the case.

    \item \textbf{Statistical bounds:} We can treat the click count $k_j$ of each detector as a realization of a binomial random variable --- note this is because we have restricted the detector behaviour to be IID within our universe of models. For such variables, one can apply a standard IID concentration inequality (e.g., the Hoeffding bound \cite[Theorem 2]{hoeffding1994probability}, or other tail bounds on the binomial distribution) to calculate a finite-size penalty term $\delta$, with the following property: the probability that the random variable $f_j + \delta$ lies below the true dark count rate of that detector is at most $\ecert$. Note that $\delta$ will depend on the number of time-bins $\nChar$ and the desired $\ecert$ value. 
    
    With this, we straightforwardly have the fact that the probability that the random variable 
    \begin{align}
    \max_j \left\{f_j + \delta\right\} 
    \end{align}
    lies below the true maximum dark count rate over all detectors is also at most $\ecert$.\footnote{To see this: note that this event is equivalent to \emph{all} $f_j + \delta$ values lying below the true maximum dark count rate. Now focusing on the detector with the highest true dark count rate, this implies in particular that the $f_j + \delta$ value for that detector lies below its true dark count rate, and thus this event has probability at most $\ecert$.} (In the language of confidence intervals: $\left[0 , \max_j \left\{f_j + \delta\right\} \right]$
    is a confidence interval for the maximum dark count rate over all detectors, where we have set the left endpoint to the ``trivial'' extremal value of $0$ for a dark count rate.) 
    Using this, we can explicitly set the approval condition: the certification procedure approves the device if and only if the observed frequencies satisfy 
    \begin{equation} \label{eq:appCond}
        \max_j \left\{f_j + \delta \right\}< \xdark^\mathrm{upp} .
    \end{equation}
    From the preceding discussion, we have that if the true dark count rate of any detector is larger than $\xdark^\mathrm{upp}$, then the probability of the frequencies $f_j$ fluctuating low enough to satisfy Eq. (\ref{eq:appCond}) is at most $\ecert$.
\end{enumerate}

\begin{remark}
    The certification procedure presented here is straightforward only because of the large number of simplifying assumptions made in the physical model. More realistic device models with afterpulsing, non-unit efficiencies, dead times, etc.\ would require a more involved certification procedure which we leave for future work.
\end{remark}

In the above example, we framed the approval condition in terms of a single random variable $\max_j \left\{f_j + \delta \right\}$, rather than 4 separate conditions that each of the random variables $f_j + \delta$ is at most $\xdark^\mathrm{upp}$. While the latter is logically equivalent, we chose to present the former for several reasons. Firstly, this directly matches the framing for single-parameter confidence intervals in Sec.~\ref{sec:certify}; secondly, it  demonstrates constructing confidence intervals for parameters that do not literally correspond exactly to the expected value of some IID process (in this case, the \emph{maximum} dark count rate over the 4 detectors). Such a perspective may be useful in other contexts, for instance when considering non-constant behaviour as mentioned in Sec.~\ref{sec:certify}.

\subsection{Robust QKD security proof} \label{subsec:robQKDProofDCs}

For the QKD security proof, we can appeal to existing works \cite{nahar2025imperfect,nahar2026proof} that can prove security with an upper bound on the dark count rate. Thus, for every device model $\mdl \in \universe$ with parameter $\xdark<\xdark^\mathrm{upp}$, the $j^\text{th}$ QKD protocol can be proved to be $\esecure_j$-secure under device model $\mdl$ (Definition~\ref{def:modelsecure}) by a direct application of Refs.~\cite{nahar2025imperfect,nahar2026proof} (for a suitable choice of the final key length, which depends on the value $\xdark^\mathrm{upp}$). This is exactly criterion~\ref{crit:QKD} defined in Proposition~\ref{prop:mainbound}.

\subsection{Application of the framework}

This procedure directly satisfies the requirements of Proposition~\ref{prop:mainbound}. The approval condition in the previous subsection is constructed specifically such that if the maximum dark count rate over all detectors $\xdark$ were greater than $\xdark^{\mathrm{upp}}$ (i.e., the parameter was outside $\robustset$), then the probability of the certification procedure approving the device is upper bounded by $\ecert$. This corresponds exactly to criterion~\ref{crit:cert} defined in Proposition~\ref{prop:mainbound}. Similarly, as described in Sec.~\ref{subsec:robQKDProofDCs}, a device model $\mdl\in\universe$ with parameter $\xdark\in\robustset$ can be shown to satisfy criterion~\ref{crit:QKD} defined in Proposition~\ref{prop:mainbound}.

As a result, we can conclude from Proposition~\ref{prop:mainbound} that the joint process involving characterization followed by $n$ QKD protocols (with the keys potentially being used between each QKD protocol) is ``$\etotal_n$-secure'' in the sense formalized in Eq.~(\ref{eq:mainbound}).

}

\section{Conclusion}
\label{sec:conclusion}

Overall, in this work we have formalized a framework that allows us to draw rigorous conclusions regarding QKD device certification, as long as the certification procedure and QKD security proofs satisfy some criteria, via Proposition~\ref{prop:mainbound}. In particular, this establishes those criteria as a useful point to focus on when considering those aspects, providing guiding principles for future considerations of device certification procedures.

Of course, a valid consideration is the question of whether there can be potential alternatives to those criteria. One natural possibility is to observe that it may be too restrictive to focus only on certification procedures that make a ``binary'' decision of whether to approve or reject the devices. More flexibly, we should consider the possibility that after the certification procedure constructs confidence intervals for the parameters $\symvec{\param}$, these confidence intervals are used to make ``adaptive'' choices of the QKD protocols performed using the device, for instance by having the protocols output keys of shorter length if the confidence intervals suggest that the device imperfections are ``larger''. This could perhaps be more accurately referred to as a ``characterization'' procedure, rather than ``certification''. (This distinction is similar in spirit to the difference between fixed-length and variable-length QKD protocols~\cite{PR22}.) Using the ideas presented in this work, we can indeed obtain one possible rigorous formalization of this idea, which we present in Appendix~\ref{app:varproof}; we defer a more detailed analysis to future work.

Finally, we reiterate that while Proposition~\ref{prop:mainbound} is presented in terms of a trace-distance bound, we believe it is possible to connect this to the framework of Abstract Cryptography~\cite{MR11,PR22} under some technical assumptions, hence yielding a composable notion of security. While it is often said that the standard QKD trace-distance security criterion~\eqref{eq:standarddefn} ``implies composable security'', this property is not an immediate consequence of an arbitrary trace-distance security definition, but rather it must be proven within a suitable formalization of the notion of ``composable security''. Hence when we additionally introduce a certification procedure as in this work, it is necessary to re-evaluate whether an analogous claim is valid.
We outline some relevant points regarding this in Appendix~\ref{app:AC}, though a fully formal claim would require addressing certain technicalities regarding how to model QKD devices as ``resources'' within the Abstract Cryptography framework.

\section*{Acknowledgements}
We thank Jan Krause, Norbert L\"{u}tkenhaus, and Jerome Wiesemann for helpful discussions.
This work was performed at the Institute for Quantum Computing, at the University of Waterloo, which is supported by Innovation, Science, and Economic Development Canada. This work was supported by NSERC under the Discovery Grants Program, Grant No. 341495.
This research has been supported by Alliance QUINT.

\printbibliography

\appendix

\section{Proof of main proposition}
\label{app:mainproof}

For clarity, we first re-state the claim of the main proposition:
\mainbound*
\begin{proof}
    We prove the proposition by considering two exhaustive cases.
    First consider the case of a model $\mdl\in\universe$ with parameters $\symvec{\param}\notin\robustset$. Then criterion \ref{crit:cert} states that $\Pr[F=1] \leq \ecert$. Since the trace distance between two quantum states is upper bounded by 1, we get 
    \begin{align}
        \Pr[F=1]\; d\left( \channAMF_n \circ \dots \channAMF_1 \left[\sigma_{Q_0 | F=1}\right] \,,\, \mathcal{K}_n \circ \channAMF_n \circ \dots \mathcal{K}_1 \circ \channAMF_1 \left[\sigma_{Q_0 | F=1}\right] \right)\leq \ecert \leq \etotal_n,
    \end{align}
    as required.
    
    Next consider the case of a model $\mdl\in\universe$ with parameters $\symvec{\param}\in\robustset$. Observe that criterion \ref{crit:QKD} implies that for all $j$ and any attack $\atk$,
    \begin{align} \label{eq:epsSecurityjDefinition}
        \sup_{\sigma^{(j)}_{Q_{j-1}}} d\left( \channAMF_j\left[\sigma^{(j)}_{Q_{j-1}}\right] \,,\, \mathcal{K}_j \circ \channAMF_j\left[\sigma^{(j)}_{Q_{j-1}}\right] \right) \leq \esecure_j,
    \end{align} 
    with the supremum being taken over all quantum states $\sigma^{(j)}_{Q_{j-1}}$.
    We use this to first show that for the $n=2$ case, we have $d\left( \channAMF_2 \circ \channAMF_1 \left[\sigma_{Q_0 | F=1}\right] \,,\, \mathcal{K}_2 \circ \channAMF_2 \circ \mathcal{K}_1 \circ \channAMF_1 \left[\sigma_{Q_0 | F=1}\right] \right)\leq \esecure_1+\esecure_2$, as follows:
    \begin{align}
        \nonumber &d\left( \channAMF_2 \circ \channAMF_1 \left[\sigma_{Q_0 | F=1}\right] \,,\, \mathcal{K}_2 \circ \channAMF_2 \circ \mathcal{K}_1 \circ \channAMF_1 \left[\sigma_{Q_0 | F=1}\right] \right)\\
        \nonumber = &\frac{1}{2}\norm{\left(\channAMF_2 \circ \channAMF_1 - \mathcal{K}_2 \circ \channAMF_2 \circ \mathcal{K}_1 \circ \channAMF_1\right)\left[\sigma_{Q_0 | F=1}\right]}_1\\
        \nonumber = &\frac{1}{2}\norm{\left(\channAMF_2 \circ \channAMF_1 -\channAMF_2 \circ\mathcal{K}_1 \circ \channAMF_1+\channAMF_2 \circ\mathcal{K}_1 \circ \channAMF_1- \mathcal{K}_2 \circ \channAMF_2 \circ \mathcal{K}_1 \circ \channAMF_1\right)\left[\sigma_{Q_0 | F=1}\right]}_1\\
        \label{eq:triangleIneqUsed}\leq &\frac{1}{2}\norm{\channAMF_2\circ\left(\channAMF_1-\mathcal{K}_1 \circ \channAMF_1\right)\left[\sigma_{Q_0 | F=1}\right]}_1 + \frac{1}{2}\norm{\left(\channAMF_2-\mathcal{K}_2 \circ \channAMF_2\right)\circ\channAMF_1\left[\sigma_{Q_0 | F=1}\right]}_1\\
        \label{eq:NonIncreasingUnderChannelUsed}\leq &\frac{1}{2}\norm{\left(\channAMF_1-\mathcal{K}_1 \circ \channAMF_1\right)\left[\sigma_{Q_0 | F=1}\right]}_1 + \frac{1}{2}\norm{\left(\channAMF_2-\mathcal{K}_2 \circ \channAMF_2\right)\circ\channAMF_1\left[\sigma_{Q_0 | F=1}\right]}_1\\
        \nonumber \leq &\sup_{\sigma^{(1)}_{Q_{0}}}\frac{1}{2}\norm{\left(\channAMF_1-\mathcal{K}_1 \circ \channAMF_1\right)\left[\sigma^{(1)}_{Q_{0}}\right]}_1 + \sup_{\sigma^{(2)}_{Q_{1}}}\frac{1}{2}\norm{\left(\channAMF_2-\mathcal{K}_2 \circ \channAMF_2\right)\left[\sigma^{(2)}_{Q_{1}}\right]}_1\\
        \label{eq:EpsSecurityUsed} \leq &\esecure_1+ \esecure_2,
    \end{align}
    where Eq. (\ref{eq:triangleIneqUsed}) follows from the triangle inequality of the $1$-norm, Eq. (\ref{eq:NonIncreasingUnderChannelUsed}) follows from the fact that the $1$-norm is non-increasing under the action of channels, and Eq. (\ref{eq:EpsSecurityUsed}) follows from Eq. (\ref{eq:epsSecurityjDefinition}).
    The general case $$d\left( \channAMF_n \circ \dots \channAMF_1 \left[\sigma_{Q_0 | F=1}\right] \,,\, \mathcal{K}_n \circ \channAMF_n \circ \dots \mathcal{K}_1 \circ \channAMF_1 \left[\sigma_{Q_0 | F=1}\right] \right)\leq \sum_{j=1}^n \esecure \leq \etotal_n$$ follows by iterating the same argument $n$ times, yielding the claimed result.
    
    Note that the above proof also works identically if instead we set $\etotal_n = \max\left\{\ecert \,,\, \sum_{j=1}^n \esecure\right\}$, as claimed.
\end{proof}

\section{Adaptive protocols based on ``characterization'' rather than ``certification''}
\label{app:varproof}
\newcommand{\estreg}{\Lambda}
\newcommand{\estval}{\lambda}
\newcommand{\certset}{\Lambda}
\newcommand{\certsetval}{\lambda}

\newcommand{\channAMC}{\mathcal{E}^{\atk,\mdl,\estval}}
\newcommand{\channAMCF}{\overline{\mathcal{E}}^{\atk,\mdl,\estval}}

The results in this section are more succinctly described by introducing the concept of confidence regions, which include confidence intervals as a special case:
\begin{definition}\label{def:confregion}
Consider a probability distribution described by some tuple of parameters $\symvec{\param}$. 
A \term{confidence region} for $\symvec{\param}$ at confidence level $\eps$ is a (set-valued) random variable  $\regconf$ such that $\regconf$ contains $\symvec{\param}$ with probability at least $1-\eps$.
\end{definition}

With this in mind, we can more rigorously describe the process of how a characterization procedure should ``estimate'' the device parameters using confidence regions, followed by performing QKD protocol instances that explicitly depend on the obtained confidence regions.

\begin{restatable}{proposition}{varversion}\label{prop:varversion}
Consider a device which undergoes a characterization procedure followed by up to $\Nmax$ QKD protocol instances (which can depend on the output of the characterization procedure), such that the following criteria are satisfied:
\begin{enumerate}
\item For any device model $\mdl \in \universe$, with parameters $\symvec{\param}$, the characterization procedure outputs a confidence region for $\symvec{\param}$ at confidence level $\ecert$. We shall denote this confidence region as $\certset$, and in the subsequent statements we shall also freely treat it as a classical register, via the standard identification between random variables and classical registers.\footnote{For the purposes of this work, we shall suppose that $\estreg$ can only take values in a countable set, to avoid technical issues with uncountably-infinite sets. This should not be a significant restriction in practice, since real-life characterization procedures can typically only output values belonging to some countable set (for instance, finite-precision numbers of some form), even when estimating real-valued parameters.} 

\item Conditioned on the value $\estval$ stored in $\estreg$, the QKD protocol instances satisfy the following criterion: for any device model $\mdl$ with parameters $\symvec{\param} \in \certsetval$, 
the $j^\text{th}$ QKD protocol instance is $\esecure_j$-secure under device model $\mdl$ (Definition~\ref{def:modelsecure}). For full clarity, in the subsequent expressions we will instead write the channels in that definition as $\channAMC_j$, explicitly denoting their dependence on $\estval$ (since the protocols being performed can depend on $\estval$).
\end{enumerate}
Then for any $n \in\{1,2,\dots,\Nmax\}$, any state $\sigma_{Q_0 \estreg}$ at the start of the first protocol instance (after the characterization procedure), any attack $\atk$ and device model $\mdl \in \universe$, and any sequence of channels $\channF_j: \KA_j \KB_j E_j \to Q_j$ for $j\in\{1,2,\dots,n-1\}$, we have: \begin{align}\label{eq:varbound}
\begin{gathered}
\sum_{\estval} \Pr[\estreg=\estval]\; d\left( \channAMCF_n \circ \dots \channAMCF_1 \left[\sigma_{Q_0 | \estval}\right] \,,\, \mathcal{K}_n \circ \channAMCF_n \circ \dots \mathcal{K}_1 \circ \channAMCF_1 \left[\sigma_{Q_0 | \estval}\right] \right)\leq \etotal_n, \\
\text{with } \etotal_n = 
\ecert+ \sum_{j=1}^n \esecure_j ,
\end{gathered}
\end{align}
where for notational compactness we write $\channAMCF_j \defvar \channAMC_j \circ \channF_{j-1}$ (defining $\channF_0$ to be the identity channel on $Q_0$), and $\mathcal{K}_j$ is a channel that reads the length $\ell$ of the keys on registers $\KA_j \KB_j$ and replaces them with perfectly correlated uniform keys \emph{of that length}, i.e.~the state $\omega^\ell$ in~\eqref{eq:idealkeyj}.
\end{restatable}

Note that the above description implicitly incorporates as a special case the possibility that some outcome values from the characterization procedure lead to the devices being ``rejected'' for use in QKD, simply by having the subsequent ``QKD protocols'' all output zero-length keys (which trivially fulfill the $\esecure_j$-security criterion). Also,
similar to the discussion in Sec.~\ref{subsec:mainresult}, the bound~\eqref{eq:varbound} can be equivalently rewritten as
\begin{align}
d\left( \rho_{\KA_n \KB_n E_n \estreg} \,,\, \widetilde{\rho}_{\KA_n \KB_n E_n \estreg} \right)\leq \etotal_n,
\end{align}
where
\begin{align}
&\rho_{\KA_n \KB_n E_n \estreg} = \sum_{\estval} \Pr[\estreg=\estval]\;\channAMCF_n \circ \dots \channAMCF_1 \left[\sigma_{Q_0 | \estval}\right] \otimes \pure{\estval},\\
&\widetilde{\rho}_{\KA_n \KB_n E_n \estreg} = \sum_{\estval} \Pr[\estreg=\estval]\; \mathcal{K}_n \circ \channAMCF_n \circ \dots \mathcal{K}_1 \circ \channAMCF_1 \left[\sigma_{Q_0 | \estval}\right]  \otimes \pure{\estval}.
\end{align}
Viewed this way, Proposition~\ref{prop:varversion} is essentially stating that the \emph{mixture} of states produced by averaging over the outcomes of the characterization procedure (i.e.~$\rho_{\KA_n \KB_n E_n \estreg}$) is close in trace distance to a somewhat more ``ideal'' state $\widetilde{\rho}_{\KA_n \KB_n E_n \estreg}$,  which is similar to $\rho_{\KA_n \KB_n E_n \estreg}$ except that the key registers $\KA_j \KB_j$ are replaced by ideal keys of the same length immediately after each protocol instance.

We now present the proof of Proposition~\ref{prop:varversion}, which uses essentially the same ideas as the variable-length security proof in~\cite{tupkary2024security}.
\begin{proof}
Consider any device model $\mdl \in \universe$, with corresponding parameters $\symvec{\param}$.
Partition the summation in~\eqref{eq:varbound} into summations over the $\estval$ values such that $\symvec{\param} \notin \certsetval$, and the $\estval$ values such that $\symvec{\param} \in \certsetval$.
For the former, we have
\begin{align}
& \sum_{\estval \,\suchthat\, \symvec{\param} \notin \certsetval} \Pr[\estreg=\estval]\; d\left( \channAMCF_n \circ \dots \channAMCF_1 \left[\sigma_{Q_0 | \estval}\right] \,,\, \mathcal{K}_n \circ \channAMCF_n \circ \dots \mathcal{K}_1 \circ \channAMCF_1 \left[\sigma_{Q_0 | \estval}\right] \right) \nonumber\\
\leq& \sum_{\estval \,\suchthat\, \symvec{\param} \notin \certsetval} \Pr[\estreg=\estval] \nonumber\\
\leq& \ecert,
\end{align}
where the last line holds by recalling that by hypothesis, $\certset$ is a confidence region for $\symvec{\param}$ at confidence level $\ecert$.  
For the latter, we have
\begin{align}&\sum_{\estval \,\suchthat\, \symvec{\param} \in \certsetval} \Pr[\estreg=\estval]\; d\left( \channAMCF_n \circ \dots \channAMCF_1 \left[\sigma_{Q_0 | \estval}\right] \,,\, \mathcal{K}_n \circ \channAMCF_n \circ \dots \mathcal{K}_1 \circ \channAMCF_1 \left[\sigma_{Q_0 | \estval}\right] \right) \nonumber\\
\leq& \sum_{j=1}^n \esecure_j ,
\end{align}
by the same reasoning as in Appendix~\ref{app:mainproof}, noting that for each $\estval$ value in the sum, the trace-distance term is evaluated for QKD protocol instances that are $\esecure_j$-secure under this device model $\mdl$ by hypothesis (since $\symvec{\param} \in \certsetval$). Adding these bounds yields the claimed final result.
\end{proof}

\section{Composability under the Abstract Cryptography framework}
\label{app:AC}

This appendix is written assuming reader familiarity with the Abstract Cryptography (AC) framework presented in~\cite{PR22}. (An earlier version of the framework can be found in~\cite{MR11}, but there are some differences in the formalism.)

We shall discuss what can be stated in the AC framework (with respect to the version described in~\cite{PR22}) regarding an ``overall'' protocol consisting of a certification procedure that either approves or rejects a device, followed by $\Nmax$ QKD protocol instances if it approves. 
For this discussion, we shall require that when the certification procedure rejects, the devices are never used for any QKD protocols.

\begin{remark}
Ideally, one might wish to instead prove some sort of composable security statement regarding the certification procedure as a ``standalone'' protocol, then combine it with security statements for the individual QKD protocol instances. However, it appears difficult to define a suitable notion of an ideal resource to construct from the certification procedure, and we were unable to satisfactorily resolve this question. We hence leave this possibility for future work.
\end{remark}

We first recall that in the AC framework, one aims to make statements that a protocol constructs some resource from another resource within some $\eps$. To connect this to the formalism of device models described in Section \ref{sec:framework}, let us impose the following condition on device models, focusing on prepare-and-measure QKD (the modifications for entanglement-based or measurement-device-independent QKD are straightforward, by changing which parties hold preparation and/or measurement devices). {We note that the statement of this condition is slightly informal, as it is not entirely clear whether the resources outlined within it can be completely formalized within the existing AC framework --- however, we believe that some version of such formalization should be possible, and hence the remainder of this appendix will proceed under the assumption that this can be done.}
\begin{condition} \label{cond:ACresource}
We suppose that for every device model $\mdl\in\universe$, it is possible to define a corresponding pair of resources in the AC framework, namely a state preparation device for Alice that can repeatedly accept a classical input and generate a state based on the input (and possibly some memory it stores over time), and a measurement device for Bob that can repeatedly accept a classical input and a quantum system, then outputs the outcome of some measurement on that quantum system depending on the classical input (and possibly some memory it stores over time).\footnote{Strictly speaking, if the certification procedure involves interacting with these devices in some more general fashion, then these resources also need to be further generalized in order to allow such operations, but the overall conclusion remains.} 
We additionally require these resources to have an operation that ``clears'' any memory in the devices.
\end{condition}

Note that at the practical level, the ``memory clearing'' operation may be as simple as leaving the devices unused for some period of time, under the physical assumption that memory effects are sufficiently weak that they will dissipate given enough time. We impose this requirement simply as an approach to address the issues discussed in Remark~\ref{remark:memory}, though it might be possible to explore in future work whether there are other ways to overcome the issue.

\begin{remark}\label{remark:modelquantifier}
As a minor variant of Condition~\ref{cond:ACresource}, it might be possible to instead attempt to simply \emph{identify} the models with those resources, since at the typical level of abstraction in the AC framework, this would also suffice to capture all properties relevant to the subsequent discussion. Yet another possible variant might be to instead attempt to define a somewhat more general resource that can accept as input some specification of a device model $\mdl\in\universe$ from Eve, and then gives Alice and Bob various state preparation and measurement functionalities depending on this specification. However, we avoid doing so here in order to avoid needing to formalize the models $\mdl\in\universe$ in some fashion compatible with being inputs to a resource in the AC framework; instead, we couch our subsequent security statements in terms of a quantifier over all models $\mdl\in\universe$. (This is analogous to the discussion in~\cite[Sec.~VI.D]{PR22} for device-independent QKD, regarding whether the security statement is quantified over all device behaviours, as compared to allowing the devices to be ``universal computers'' that can have their behaviour determined by an input from Eve.)
\end{remark}

Having established this, we now need to specify the starting resources and the resource we aim to construct, for both honest and dishonest Eve~\cite{PR22}. Following the above discussion, we view these resources as being implicitly defined as functions of models $\mdl\in\universe$, and we will later present our security statements in terms of a quantifier over $\mdl\in\universe$.
First considering dishonest Eve, we take the starting resources $\mathcal{R}_\mathrm{ini}$ to be, similarly to~\cite{PR22}:
\begin{itemize}
\item An insecure quantum channel between Alice and Bob.
\item An authenticated classical channel between Alice and Bob.
\item The resource described above in Condition~\ref{cond:ACresource} (implicitly depending on the model $\mdl$).
\end{itemize}
As for the resource we aim to construct, we define it as follows (essentially a multiple-key version of the secret key resource in~\cite[Fig.~18]{PR22}), given an increasing sequence of times $t_j$ for $j\in\{1,2,\dots,\Nmax\}$:
\begin{itemize}
\item $\mathcal{R}_\mathrm{ideal}$ is a resource where at each time $t_j$, it accepts a classical input value $\ell_j \in \mathbb{N}$ from Eve's interface, and then outputs perfect shared secret keys of length $\ell_j$ to Alice and Bob's interfaces.
\end{itemize}

With this in mind, we can define the protocol we are considering, performed using the starting resources $\mathcal{R}_\mathrm{ini}$. In order to minimize difficulties regarding composability, we shall suppose that the devices (as defined in Condition~\ref{cond:ACresource}) are not used for any purposes other than this protocol --- it might be possible to do so if the device memories are cleared appropriately, but we shall not consider this in depth here.
\begin{itemize}
\item First, a certification procedure is performed on $\mathcal{R}_\mathrm{ini}$, generating either an approve or reject decision that is used to determine the next action.
\item If the certification procedure approves, the resources $\mathcal{R}_\mathrm{ini}$ are then used for $\Nmax$ QKD protocol instances that output keys at the times $t_j$, under the constraint after each QKD protocol instance, \emph{the device memories are cleared before they are used again}. If the certification procedure rejects, the devices are discarded, which we shall formalize by having the protocol output zero-length keys at all the times $t_j$.
\end{itemize}
In the AC framework as described in~\cite{PR22}, the above protocol would be a converter applied to Alice and Bob's interfaces in $\mathcal{R}_\mathrm{ini}$, with the inner interface connecting to $\mathcal{R}_\mathrm{ini}$, and the outer interface producing keys at times $t_j$. (Note that the approve/reject decision from the certification step is to be interpreted as an ``internal process'' of the protocol and is not released at any of the interfaces.) 

With this, we now state our main claim {under the assumption that Condition~\ref{cond:ACresource} can be properly formalized, followed by an outline of how it would be proven given a suitable formalization of that condition}. Since distinguishability in the AC framework is always to be considered with respect to some class of distinguishers, we phrase our following result as being implicitly in terms of the class of distinguishers. Also, as discussed in Remark~\ref{remark:modelquantifier}, we simply state the result with a quantifier over all models $\mdl\in\universe$, rather than attempting to build the dependence on the model $\mdl$ ``into'' the resource in some fashion.
\begin{proposition}
Consider any device model $\mdl\in\universe$ satisfying Condition~\ref{cond:ACresource} and take any increasing sequence of times $t_j$. Let $\mathcal{R}_\mathrm{ini}$ and $\mathcal{R}_\mathrm{ideal}$ be the resources defined above, for the case of dishonest Eve, and consider the protocol defined above. 
Suppose furthermore that this protocol satisfies the two criteria in Proposition~\ref{prop:mainbound} for some class of attacks $\atk$, and define 
\begin{align}
\etotal \defvar \ecert+ \sum_{j=1}^{\Nmax} \esecure_j .
\end{align}
Then this protocol constructs the resource $\mathcal{R}_\mathrm{ideal}$ from $\mathcal{R}_\mathrm{ini}$ within $\etotal$ in the sense of~\cite[Definition~1]{PR22}, as long as the class of attacks $\atk$ is large enough to describe all distinguisher strategies considered in that definition. 
\end{proposition}
\begin{proof}[Proof outline]
The idea is basically identical to~\cite[Sec.~III.B.2]{PR22}; here we outline the main structure. By definition, we need to construct a suitable simulator $\sigma_E$ to attach to Eve's interface in $\mathcal{R}_\mathrm{ideal}$, such that $\mathcal{R}_\mathrm{ideal} \sigma_E$ is $\etotal$-indistinguishable from $\pi_{AB}\mathcal{R}_\mathrm{ini}$ where $\pi_{AB}$ denotes the protocol described above performed by Alice and Bob. We follow a basically identical construction to~\cite[Sec.~III.B.2]{PR22} for this simulator. Specifically, we define it to simply run an ``internal instance'' of $\pi_{AB}\mathcal{R}_\mathrm{ini}$, so its ``outer'' interface interacts accordingly with Eve. As for how it interacts on its ``inner'' interface (i.e.~the Eve interface on $\mathcal{R}_\mathrm{ideal}$), we define it as follows: for each time $t_j$, if the certification step rejected in the ``internal instance'' of $\pi_{AB}\mathcal{R}_\mathrm{ini}$, then the simulator simply sets $\ell_j=0$, otherwise (if the certification step approved) the simulator sets $\ell_j$ equal to the length of the key produced from the $j^\text{th}$ QKD protocol in the ``internal instance'' of $\pi_{AB}\mathcal{R}_\mathrm{ini}$.

Observe that $\pi_{AB}\mathcal{R}_\mathrm{ini}$ is perfectly indistinguishable from $\mathcal{R}_\mathrm{ideal} \sigma_E$ up until time $t_1$ (as also noted in~\cite[Sec.~III.B.2]{PR22}); there are absolutely no differences in their behaviour up to that point. At time $t_1$, the first potential discrepancy occurs, since $\pi_{AB}\mathcal{R}_\mathrm{ini}$ outputs the actual key at Alice and Bob's interfaces, whereas $\mathcal{R}_\mathrm{ideal} \sigma_E$ outputs a ``replaced'' version of the key. Let us write the registers held by the distinguisher at that point as $\KA_1 \KB_1 E_1$, where $E_1$ denotes all registers held by the distinguisher just before $\KA_1 \KB_1$ were released. 
With this, observe that the states held by the distinguisher at this time in the two scenarios are respectively the reduced states (on $\KA_1 \KB_1 E_1$) of states with the form
\begin{align}
\begin{aligned}
\rho_{\KA_1 \KB_1 E_1 F} &= \Pr[F=0]\, \rho_{\KA_1 \KB_1 E_1 | F=0} \otimes \pure{0}_F +  \Pr[F=1]\, \channAM_1 \left[\sigma_{Q_0 | F=1}\right] \otimes \pure{1}_F, \\
\widetilde{\rho}_{\KA_1 \KB_1 E_1 F} &= 
\mathcal{K}_1 \left[\rho_{\KA_1 \KB_1 E_1 F}\right] \\
&= \Pr[F=0]\, 
\rho_{\KA_1 \KB_1 E_1 | F=0} \otimes \pure{0}_F +  \Pr[F=1]\, \mathcal{K}_1 \circ \channAM_1 \left[\sigma_{Q_0 | F=1}\right] \otimes \pure{1}_F,
\end{aligned}
\end{align}
where $\mathcal{K}_1, \channAM_1, \sigma_{Q_0 F}$ are as in the Proposition~\ref{prop:mainbound} notation (for some $\atk$ corresponding to the distinguisher's strategy), and to obtain the last equality we recall that in the conditional state $\rho_{\KA_1 \KB_1 E_1 | F=0}$, the key registers $\KA_1 \KB_1$ are fixed to zero-length keys, and therefore $\mathcal{K}_1$ leaves it unchanged. 
Furthermore, by Proposition~\ref{prop:mainbound} we have
\begin{align}
\begin{gathered}
\Pr[F=1]\; d\left( \channAM_1 \left[\sigma_{Q_0 | F=1}\right] \,,\, \mathcal{K}_1 \circ \channAM_1 \left[\sigma_{Q_0 | F=1}\right] \right)\leq \ecert+ \esecure_1 ,
\end{gathered}
\end{align}
and thus we have
\begin{align}\label{eq:t1dist}
d\left( \rho_{\KA_1 \KB_1 E_1 F} \,,\, \widetilde{\rho}_{\KA_1 \KB_1 E_1 F} \right) \leq \ecert+ \esecure_1 \leq \etotal,
\end{align}
so the states held by the distinguisher at time $t_1$ (on $\KA_1 \KB_1 E_1$) are $\etotal$-indistinguishable. 

Proceeding on to later times, recall we suppose that after the first QKD protocol finishes, the devices have their memories cleared before they are used again in any form. Therefore, without loss of generality we can suppose any memory registers in the devices are traced out immediately after $t_1$ (since they cannot have any operational impact if the devices are not used before these memories are cleared). Under this picture, observe that the registers $\KA_1 \KB_1 E_1 F$ suffice to capture the \emph{entire} state of the resources and the distinguisher at that point. As we have shown above in~\eqref{eq:t1dist}, the states on these registers are $\etotal$-indistinguishable in the two scenarios. From this, we can conclude that up until time $t_2$ (when the keys from the second QKD protocol are produced), the distinguishing advantage remains at most $\etotal$, because all further states held by the distinguisher up until that time can then be described via some completely positive, trace-preserving (CPTP) map acting on those states, which does not increase trace distance. 

Then for the states at time $t_2$, we apply the same argument structure as above to conclude the trace distance is again at most $\ecert+ \esecure_1 + \esecure_2 \leq \etotal$. Continuing similarly for the remainder of the protocol, we conclude $\pi_{AB}\mathcal{R}_\mathrm{ini}$ is and $\mathcal{R}_\mathrm{ideal} \sigma_E$ are always $\etotal$-indistinguishable, yielding the desired result.
\end{proof}

\begin{remark}
The ``adaptive'' device-characterization version described in Appendix~\ref{app:varproof} would also construct the same resource within $\etotal$ (under the same informal conditions discussed above), via an analogous argument. Specifically, rather than having the simulator set $\ell_j=0$ for all $j$ whenever the certification step aborts, we simply have it supply the length $\ell_j$ of the key  produced from the $j^\text{th}$ QKD protocol in the ``internal instance'' of $\pi_{AB}\mathcal{R}_\mathrm{ini}$ (in this version, these lengths will also implicitly depend on the outcome from the certification step). The rest of the proof outline carries through in the same manner.
\end{remark}

We now turn to the case where Eve is honest, as is required for a full analysis under the AC framework. The following discussion is to be understood to apply both for the device certification procedure in Section \ref{sec:framework} and the  ``adaptive'' device-characterization version described in Appendix~\ref{app:varproof}, as the same claims apply to both. Again following~\cite{PR22}, we take the starting resources to be the same as above, except that the insecure quantum channel outputs nothing on Eve's interface and is taken to have some fixed ``honest'' noise level between Alice and Bob. Similarly, the resource we aim to construct is one that outputs nothing at Eve's interface, and at each time $t_j$, it outputs a mixture of perfect shared secret keys of different lengths to Alice and Bob's interfaces, such that the distribution of lengths matches that which would be obtained from the honest behaviour of the insecure quantum channel. The analysis of this case is fairly straightforward, following the same arguments in~\cite[Sec.~III.B.4]{PR22}; we again conclude that for any $\mdl\in\universe$ the protocol indeed constructs the claimed resource within $\etotal$. Note however that there might be many $\mdl\in\universe$ such that the resulting constructed resource in this case is not of practical interest (e.g.~one that almost never produces any keys, because for that model $\mdl$ the certification procedure almost always rejects). Hence if desired, we can instead focus only on a specific model $\mdl_\mathrm{hon}\in\universe$ corresponding to some ``honest'' device model that yields a practically interesting final resource.

\end{document}